\documentclass[12pt]{amsart}
\usepackage{amsmath}
\usepackage{amsfonts}
\usepackage{amssymb}

\newtheorem{lemma}{Lemma}[section]
\newtheorem{theorem}[lemma]{Theorem}
\newtheorem{coro}[lemma]{Corollary}
\newtheorem{prop}[lemma]{Proposition}

\newtheorem{rem}[lemma]{Remark}
\newtheorem{definition}[lemma]{Definition}

\newcommand{\Rdst}{{\mathbb{R}^d}}

\newcommand{\Rtdst}{{\mathbb{R}^{2d}}}

\thanks{M.\ d.\ G.\ was supported by the grant P 23902 from the Austrian
  Science Fund (FWF). K.\ G.\ was
  supported in part by the  project P26273-N25 of the
Austrian Science Fund (FWF).
\mbox{J. L. R.} gratefully acknowledges support from a Marie
Curie fellowship, within the 7th. European 
Community Framework program, under grant PIIF-GA-2012-327063.}

\begin{document}
\author{Maurice A. de Gosson}
\email{maurice.de.gosson@univie.ac.at }
\author{Karlheinz Gr\"ochenig}
\email{karlheinz.groechenig@univie.ac.at}
\author{Jos\'e Luis Romero}
\email{jose.luis.romero@univie.ac.at}
\address{Faculty of Mathematics \\
University of Vienna \\
Oskar-Morgenstern-Platz 1 \\
A-1090 Vienna, Austria}
\title{Stability of Gabor frames under small time Hamiltonian evolutions}



\begin{abstract}
We consider Hamiltonian deformations of Gabor systems, where the window
evolves according to the action of a Schr\"odinger propagator and the
phase-space nodes evolve according to the corresponding Hamiltonian flow. We
prove the stability of the frame property for small times and Hamiltonians
consisting of a quadratic polynomial plus a potential in the Sj\"ostrand
class with bounded second order derivatives. This answers a question raised
in [de Gosson, M. Symplectic and Hamiltonian Deformations of Gabor Frames.
Appl. Comput. Harmon. Anal. Vol.\ 38 No.2, (2015) p.196--221.]
\end{abstract}

\maketitle

\section{Introduction}

Let $\mathcal{H}(x,p)$ be a Hamiltonian on ${\mathbb{R}^{2d}}$ and
$H:=\mathcal{H}^w$ its Weyl quantization. The solution to the Schr\"odinger
equation 
\begin{align*}
&i \partial_t u(t,\cdot) = H u(t,\cdot), \qquad t\in {\mathbb{R}}, \\
&u(0,\cdot)=f,
\end{align*}
is given by the propagation formula $u(t,\cdot)=e^{-i t H}f$. The model case
is the one of a real quadratic (homogeneous) Hamiltonian: $\mathcal{H}(x,p)
= \left<M (x,p),(x,p)\right> $, with $M \in {\mathbb{R}}^{2d\times 2d}$
symmetric. In this case, the evolution operator $e^{-i t H}$ and the
(symmetric) time-frequency shift operators 
\begin{align}  \label{eq_tfs}
\rho(z) f := e^{-\pi i x\xi} e^{2 \pi i \xi} f(\cdot-x), \qquad z=(x,\xi)
\in {\mathbb{R}^{2d}},
\end{align}
satisfy the symplectic covariance relation 
\begin{align}  \label{eq_cov}
e^{-itH} \rho(z) = \rho(e^{2tJM} z) e^{-itH},
\end{align}
where $J=\big(
\begin{smallmatrix}
  0 & I \\ -I & 0
\end{smallmatrix}\big)
$ is the standard symplectic form (see for example \cite[Chapter 15]{dego11}). Thanks to \eqref{eq_cov}, the action of 
the evolution operator on a state 
$f$, can be understood by considering an expansion into coherent states: 
\begin{align}  \label{eq_exp}
f = \sum_{\lambda \in \Lambda} c_\lambda \rho(\lambda) g,
\end{align}
where $g$ is a smooth, fast decaying function, $\Lambda \subseteq {\mathbb{R}^d}$ is a set of phase-space nodes and 
$c_\lambda \in {\mathbb{C}}$. Such an
expansion is a discrete version of the continuous coherent state
representation \cite{alanga00},  and the canonical choice for $g$ is a
Gaussian function. The evolution generated by the quadratic Hamiltonian $\mathcal{H}$ is then given by 
\begin{align}  \label{eq_ev}
e^{-i t H} f = \sum_{\lambda \in \Lambda} c_\lambda \rho(e^{2tJM} \lambda)
e^{-i t H} g,
\end{align}
and therefore the description of the evolution of an arbitrary state $f$ is
reduced to the one of $g$. If $\mathcal{H}(x,p) = x^2 + p^2$ is the
harmonic oscillator and  $g$ is chosen to
be an adequate Gaussian function, then $e^{-i t H}g=g$, and \eqref{eq_ev}
amounts to a rearrangement of the time-frequency content of $f$. The case of
higher order Hermite functions is also important since these correspond to
higher energy Landau levels (see \cite{MR3203099}).

The collection of coherent states 
\begin{align*}
\mathcal{G}(g, \Lambda) := \big\{ \, \rho(\lambda)g \, : \, \lambda \in
\Lambda \, \big\}
\end{align*}
is called a Gabor system, and it is a \emph{Gabor frame} if every $f \in L^2({\mathbb{R}^d})$ admits an expansion as in 
\eqref{eq_exp} with $\lVert
c\rVert_2 \asymp \lVert f\rVert_2$. In this case, several properties of $f$
can be read from the coefficients $c$. The theory of Gabor frames - also
called Weyl--Heisenberg frames - plays an increasingly important role in
physics; see for instance \cite{MR3317554, ams, dego11} and the references
therein.

Recently one of us started the investigation of the relation between the
theory of Gabor frames and Hamiltonian and quantum mechanics \cite{dego13}
and introduced the notion of a \emph{Hamiltonian deformation of a Gabor
system}. For a (time-independent) Hamiltonian $\mathcal{H}(x,p)$ we let $\Phi_t(x,p)$ be the flow given by the Hamilton 
equations 
\begin{equation*}
\left\{ 
\begin{array}{ll}
\dot{x} & = \mathcal{H}_p(x,p), \\ 
\dot{p} & = -\mathcal{H}_x(x,p),
\end{array}
\right. 
\end{equation*}
and let $H:=\mathcal{H}^w$ be the Weyl quantization of $\mathcal{H}$. Given
a Gabor system $\mathcal{G}(g, \Lambda)$, we consider the time-evolved
systems 
\begin{align}  \label{eq_ham_def}
\mathcal{G}_t(g, \Lambda) := \mathcal{G}(e^{-it H}g, \Phi_t\Lambda),
\end{align}
and investigate the stability of the frame property under the evolution $\mathcal{G}(g, \Lambda) \mapsto 
\mathcal{G}_t(g, \Lambda)$. 

When $\mathcal{H}
$ is a quadratic form $\mathcal{H}(x,p) = \left<M (x,p),(x,p)\right> $, its
flow is given by the linear map $\Phi_t(x,p) = e^{2t JM} (x,p)$ and \eqref{eq_cov} expresses the fact that the evolution 
operator $e^{-i t H}$
is the \emph{metaplectic operator} associated with the linear map $e^{2t JM}$. As a consequence, $\mathcal{G}_t(g, 
\Lambda)$ is the image of $\mathcal{G}(g, \Lambda)$ under the unitary map $e^{-i t H}$ and hence it enjoys the
same spanning properties (in particular, the frame property is preserved).
This observation is called the \emph{symplectic covariance} of Gabor frames 
\cite{dego13}.

For more general Hamiltonians $\mathcal{H}$, no strict covariance property
holds, and the analysis of the deformation $t\mapsto\mathcal{G}_t(g, \Lambda)
$ is difficult. In \cite{dego13}, one of us analyzed a linearized version of
this problem and established some stability estimates (see also \cite{bebuconi15} for a higher-order approximation to 
the deformation
problem).
Based on these results, \cite{dego13} conjectured that  the evolution
$\mathcal{G}(g, \Lambda) \mapsto \mathcal{G}_t(g, \Lambda)$ preserves
the frame property for more general Hamiltonians. In particular, one
would 
expect perturbations of
quadratic Hamiltonians to exhibit a certain approximate symplectic
covariance, in the form of stability of the frame property of $\mathcal{G}_t(g, \Lambda)$ for a certain range of time.

In this article we solve the deformation problem for small times. More
precisely, we consider a perturbation of a quadratic Hamiltonian by an
element of the \emph{Sj\"{o}strand class} $M^{\infty, 1}({\mathbb{R}^{2d}})$ with
bounded second order derivatives. We also consider a Gabor frame with window
in the \emph{Feichtinger algebra} $M^1({\mathbb{R}^d})$ of functions with
integrable Wigner distribution. (See Section \ref{sec_back} and \ref{sec_tf}
for precise definitions). The following is our main result.

\begin{theorem}
\label{th_small_time} Let $a$ be a real-valued, quadratic, homogeneous
polynomial on ${\mathbb{R}^{2d}}$ and let $\sigma \in M^{\infty, 1}({\mathbb{R}^{2d}})\cap C^2({\mathbb{R}^{2d}})$ have 
bounded second order derivatives.
Consider the Hamiltonian $\mathcal{H}(t,x,p) := a(x,p)+\sigma(x,p)$. Let $H
:= \mathcal{H}^w(x,D)$ be the Weyl quantization of $\mathcal{H}$ and let $(\Phi_t)_{t \in {\mathbb{R}}}$ be the flow of 
$\mathcal{H}$.

Let $g \in M^1({\mathbb{R}^d})
$ and $\Lambda \subseteq {\mathbb{R}^{2d}}$, such that
$\mathcal{G}(g,\Lambda)$ is a Gabor frame. Then there
exists $t_0 >0$ such that for all $t \in [-t_0,t_0]$, $\mathcal{G}(e^{-itH}g, \Phi_t(\Lambda))$ is a Gabor frame.
\end{theorem}

To see what is at stake, we consider once more the symplectic
covariance property~\eqref{eq_cov}. It links the classical Hamiltonian flow
$e^{2tJM}$ on phase space to the quantum mechanical evolution. If
$\mathcal{H}$ is not quadratic, then the flow $\Phi _t$ is no longer
linear, and, in general,   there is no explicit and exact formula for the quantum
mechanical evolution. We therefore have to understand   the classical
evolution of the set $\Lambda $ under $\Phi _t$ separately from the
quantum mechanical evolution of the state (window) $g$ under
$e^{-itH}$. 

The stability of the frame property of $\mathcal{G}(g, \Phi _t(\Lambda
))$ is part of  the deformation  theory  of Gabor frames. 
While there is a significant literature on the
stability of Gabor frames under linear distortions of the time-frequency
nodes $\Lambda$ (covering perturbation of lattice parameters \cite{be94,
feka04} on the one hand, and general point sets \cite{asfeka13}), only
recently  a fully non-linear deformation theory of Gabor systems was
developed in \cite{grorro15}. It turns out that the concept of
Lipschitz deformation  is precisely the right tool to treat non-linear
Hamiltonian flows, and we will use the main result of
\cite{grorro15} in a decisive manner. 

The second ingredient in Theorem~\ref{th_small_time} is the assumption
$g\in M^1(\mathbb{R}^d)$. This is an essential assumption for Gabor
frames to be useful in phase space analysis. In particular, most
stability results for  Gabor frames under perturbations of the window require
that $g\in M^1(\mathbb{R}^d)$. Outside $M^1$ one encounters quickly
pathologies~\cite{rosu15}. In regard to our problem it is therefore
important to understand whether $M^1(\mathbb{R}^d)$ is invariant under
the evolution of the Schr\"odinger equation. This is indeed the case
for certain classes of Hamiltonians~\cite{BGOR07,cogrniro13}, and will be the second
important tool used to prove Theorem~\ref{th_small_time}.

The rest of the article is organized as follows. In Section \ref{sec_back}
we provide some definitions and background results. Section \ref{sec_tools}
collects the essential tools and derives some auxiliary estimates. Finally,
the proof of Theorem \ref{th_small_time} is presented in Section \ref{sec_de}.

\section{Background}

\label{sec_back}

\subsection{Time-frequency analysis}

\label{sec_tf} Given a function $g \in L^2({\mathbb{R}^d})$, with $\lVert
g\rVert_2=1$, the \emph{short-time Fourier transform} of a function $f \in
L^2({\mathbb{R}^d})$ with respect to the window $g$ is defined as 
\begin{align}  \label{eq_def_stft}
V_g f(x,\xi) := \left<f,e^{\pi i x\xi} \rho(x,\xi) g\right> , \qquad (x,\xi)
\in {\mathbb{R}^{2d}},
\end{align}
where $\rho(x,\xi)$ is the (symmetric) time-frequency shift defined in \eqref{eq_tfs}. The function $g$ is often called 
window and the
normalization $\lVert g\rVert_2=1$ implies that 
\begin{align}  \label{eq_stft_l2}
\lVert V_g f\rVert_{{L^2({\mathbb{R}^{2d}})}}=\lVert f\rVert_{{L^2({\mathbb{R}^d})}}, \qquad f \in 
{L^2({\mathbb{R}^d})}.
\end{align}
The standard choice for $g$ is the Gaussian $\phi(x) := 2^{d/4}
e^{-\pi\left| x \right|  ^2}$. Analogously, the Feichtinger algebra,
originally introduced in \cite{fe81-2},
is defined to be 
\begin{align*}
M^1({\mathbb{R}^d}) := \big\{ \, f \in L^2({\mathbb{R}^d}) \, : \, \lVert
f\rVert_{M^1} := \lVert V_\phi f\rVert_{L^1({\mathbb{R}^{2d}})} < +\infty \, \big\},
\end{align*}
and is  used as a standard reservoir for windows $g$. Equivalently, $f \in M^1(\mathbb{R}^d)$, if the Wigner 
distribution $Wf(x,\xi ) = \int f(x+t/2) \overline{f(x-t/2)} e^{-2\pi i \xi \cdot t} \,
dt$ of $f$ is integrable on $\mathbb{R}^{2d}$. When $g \in M^1({\mathbb{R}^d})$, the map $f \mapsto V_g f$ can be 
extended beyond $L^2({\mathbb{R}^d})$.

We define the \emph{modulation spaces} as follows: fix a non-zero $g\in 
\mathcal{S} ({\mathbb{R}^d} )$ and let $1 \leq p,q \leq \infty$. Then  $M^{p,q}({\mathbb{R}^d})$ is  the class of all 
distributions $f \in \mathcal{S}^{\prime
}({\mathbb{R}^d})$ such that 
\begin{align}  \label{eq_def_mp}
\lVert f\rVert_{M^{p,q}({\mathbb{R}^d})} := \left( \int_{{\mathbb{R}^d}}
\left( \int_{{\mathbb{R}^d}} \left| V_g f(x,\xi) \right|  ^p dx
\right)^{q/p} d\xi \right)^{1/q} < \infty,
\end{align}
with the usual modification when $p$ or $q$ is $\infty$. Different choices
of non-zero windows $g \in \mathcal{S}({\mathbb{R}^d})$ yield the same space
with equivalent norms, see \cite{fe06} and \cite[Chapter 11]{gr01}. In addition, for $g \in M^1({\mathbb{R}^d})$, the short-time Fourier transform is well-defined on
all $M^{p,q}({\mathbb{R}^d})$. Originally introduced by Feichtinger in \cite{fe83-4}, modulation spaces combine smoothness and integrability conditions. In this article, we will be mainly concerned with Feichtinger's algebra $M^1(\Rdst)$, as a window class for Gabor systems, and $M^{\infty,1}(\Rtdst)$ - also known as Sj\"ostrand's class, as a symbol class for pseudodifferential operators.

\subsection{Sampling the short-time Fourier transform}

A set $\Lambda \subseteq {\mathbb{R}^d}$ is called \emph{relatively separated} if 
\begin{align}  \label{eq_rel}
\mathop{\mathrm{rel}}(\Lambda) := \sup \{ \#(\Lambda \cap (\{x\}+[0,1]^d)) :
x \in {\ \mathbb{R}^d} \} < \infty.
\end{align}
The assumption that $g \in M^1({\mathbb{R}^d})$ implies certain sampling
estimates for the short-time Fourier transform. We quote the following
standard result (see for example \cite[Chapter 13]{gr01}.)

\begin{prop}
\label{prop_win_stab} Let $g \in M^1({\mathbb{R}^d})$ and let $\Lambda
\subseteq {\mathbb{R}^{2d}}$. Then 
\begin{align*}
\left(\sum_{\lambda \in \Lambda} \left| V_g f(\lambda) \right| 
^2\right)^{1/2} \leq C \mathop{\mathrm{rel}}(\Lambda) \lVert g\rVert_{M^1}
\lVert f\rVert_2, \qquad f \in {L^2({\mathbb{R}^d})},
\end{align*}
where the constant $C$    depends  only on the dimension $d$.
\end{prop}

\subsection{Gabor frames}

Given a window $g \in M^1({\mathbb{R}^d})$ and a relatively separated set $\Lambda \subseteq {\mathbb{R}^{2d}}$, the 
collection of functions 
\begin{align*}
\mathcal{G}(g,\Lambda) := \left \{ \rho(\lambda)g : \lambda \in \Lambda
\right \} 
\end{align*}
is called the \emph{Gabor system} generated by $g$ and $\Lambda$. It is a
Gabor frame, if there exist constants $A,B>0$ such that 
\begin{align}  \label{eq_frame}
A \lVert f\rVert_2^2 \leq \sum_{\lambda \in \Lambda} \left|
\left<f,\rho(\lambda)g\right>  \right|  ^2 \leq B \lVert f\rVert_2^2,\qquad
f \in L^2({\mathbb{R}^d}).
\end{align}
The constants $A,B$ are called frame bounds for $\mathcal{G}(g,\Lambda)$. We
remark that the definition of Gabor system given here is slightly
non-standard. In signal processing, it is more common to define the
time-frequency shifts by 
\begin{align*}
\pi(z)f(t) := e^{2\pi i \xi t} f(t-x), \qquad z=(x,\xi) \in {\mathbb{R}^d}\times {\mathbb{R}^d}, t \in {\mathbb{R}^d}.
\end{align*}
Since $\pi(x,\xi) = e^{\pi i x\xi} \rho(x,\xi)$, the choice $\rho$ has no
impact on the frame inequality in \eqref{eq_frame}. Note that the sum
in \eqref{eq_frame} is the same as $\|V_gf |_\Lambda \|_2^2$.  The use
of $\rho$ instead of $\pi $  in
this article is motivated by the symplectic covariance property in \eqref{eq_cov}, which would require additional phase 
factors if $\pi$ was
used instead of $\rho$.

The following basic fact can be found for example in \cite[Theorem 1.1]{chdehe99}.

\begin{lemma}
\label{lemma_rel_sep} If $\mathcal{G}(g,\Lambda)$ is a frame, then $\Lambda$
is relatively separated.
\end{lemma}

\section{The essential tools}

\label{sec_tools}

\subsection{Schr\"{o}dinger operators on modulation spaces}

\label{sec_sch} The \emph{Weyl transform} of a distribution $\sigma \in 
\mathcal{S}^{\prime }({\mathbb{R}^d} \times {\mathbb{R}^d})$ is an operator $\sigma^w$ that is formally defined on 
functions $f:{\mathbb{R}^d} \to {\mathbb{C}}$ as 
\begin{align*}
\sigma^w (f)(x) := \int_{{\mathbb{R}^d} \times {\mathbb{R}^d}} \sigma\left(\frac{x+y}{2},\xi\right) e^{2\pi i(x-y)\xi} 
f(y) dy d\xi, \qquad x \in {\mathbb{R}^d}.
\end{align*}
The fundamental results in the theory of pseudodifferential operators
provide conditions on $\sigma$ for the operator $\sigma^w$ to be
well-defined and bounded on various function spaces. In particular,
Sj\"ostrand proved that if $\sigma \in M^{\infty, 1}({\mathbb{R}^{2d}})$,
then $\sigma^w$ is bounded on $L^2({\mathbb{R}^d})$ \cite{sj94,sj95}. See
also \cite{gr06-3,gr06} for extensions of these results to weighted symbol
classes and modulations spaces.

The following result is one of our main tools. It shows that
perturbing a quadratic Hamiltonian with a potential in the Sj\"ostrand's class $M^{\infty, 1}({\mathbb{R}^{2d}})$ gives rise to propagators that 
are strongly continuous on $M^1({\mathbb{R}^d})$.

\begin{theorem}[{\protect\cite[Theorems 1.5 and 4.1]{cogrniro13}}]
\label{th_ham_win} Let $a$ be a real-valued, quadratic, homogeneous
polynomial on ${\mathbb{R}^{2d}}$ and let $\sigma \in M^{\infty, 1}({\mathbb{\ R}^{2d}})$. Let $H := a^w(x,D) + 
\sigma^w(x,D)$. Then $e^{i t H}$ is a
strongly continuous one-parameter group of operators on $M^1({\mathbb{R}^d})$. In other words:

\begin{itemize}
\item[(a)] for all $t \in {\mathbb{R}}$, $e^{i t H}: M^1({\mathbb{R}^d}) \to
M^1({\mathbb{R}^d})$,

\item[(b)] for each $g \in M^1({\mathbb{R}^d})$, 
\begin{align}
e^{i t H} g \longrightarrow g \mbox{ in }M^1({\mathbb{R}^d}), \mbox{ as } t
\longrightarrow 0.
\end{align}
\end{itemize}
\end{theorem}

\subsection{Deformation of Gabor frames}

Our second essential tool is a description of the stability of the frame
property of a Gabor frame $\mathcal{G}(g,\Lambda)$ under small deformations
of $\Lambda$. Our general assumption is that $g \in M^1({\mathbb{R}^d})$.
(Without this assumption the frame property might be very unstable under
perturbation of $\Lambda$, even for lattices \cite{ja03-1,rosu15}).

The classical results in signal processing describe the stability of the
frame property under the so-called jitter perturbations: if $\mathcal{G}(g,\Lambda)$ $\sup _{\lambda \in \Lambda } \inf 
_{\lambda ^{\prime }\in
\Lambda ^{\prime }} |\lambda - \lambda ^{\prime }|<\epsilon$ and $\sup
_{\lambda^{\prime }\in \Lambda^{\prime }} \inf _{\lambda \in \Lambda}
|\lambda - \lambda ^{\prime }| < \epsilon $, for sufficiently small $\varepsilon$, then $\mathcal{G} (g, \Lambda 
^{\prime })$ is also a frame. A
much deeper property is the stability of the frame condition under linear
maps $\Lambda \mapsto A \Lambda$, where $A$ is a matrix that is sufficiently
close to the identity (but possibly not symplectic!). Such results have been
derived first for lattices \cite{be94, feka04} and then for general sets 
\cite{asfeka13}. In order to deal with Hamiltonian flows, we will resort to
a recent fully non-linear stability theory \cite{grorro15}.

Let $\Lambda \subseteq {\mathbb{R}^d}$ be a set. We consider a sequence $\left \{ \Lambda_n: n \geq 1 \right \} $ of 
subsets of ${\mathbb{R}^d}$
produced in the following way. For each $n \geq 1$, let $\tau_n: \Lambda \to 
{\mathbb{R}^d}$ be a map and let $\Lambda_n := \tau_n(\Lambda) = \left \{
\tau_n(\lambda): \lambda \in \Lambda \right \} $. We assume that $\tau_n(\lambda) \longrightarrow \lambda$, as $n 
\longrightarrow \infty$, for
all $\lambda \in \Lambda$. The sequence of sets $\left \{ \Lambda_n: n \geq
1 \right \} $ together with the maps $\left \{ \tau_n: n \geq 1 \right \} $
is called a \emph{deformation} of $\Lambda$. We think of each sequence of
points $\left \{ \tau_n(\lambda): n \geq 1 \right \} $ as a (discrete) path
moving towards the endpoint $\lambda$.

We will say that $\left \{ \Lambda_n: n \geq 1 \right \} $ is a deformation
of $\Lambda$, with the understanding that a sequence of underlying maps $\left \{ \tau_n: n \geq 1 \right \} $ is also 
given.

We now describe a special class of deformations.

\begin{definition}
\label{def_lip} A deformation $\left \{ \Lambda_n: n \geq 1 \right \} $ of $\Lambda $ is called \emph{Lipschitz}, 
denoted by $\Lambda_n \xrightarrow{Lip}
\Lambda$, if the following two conditions hold:

\textbf{(L1) } Given $R>0$, 
\begin{align*}
\sup_{ \overset{\lambda, \lambda^{\prime }\in \Lambda}{\left|
\lambda-\lambda^{\prime }\right| \leq R}} \left| (\tau_n(\lambda) -
\tau_n(\lambda^{\prime })) - (\lambda - \lambda^{\prime }) \right|
\rightarrow 0, \quad \mbox {as } n \longrightarrow \infty.
\end{align*}

\textbf{(L2) } Given $R>0$, there exists $R^{\prime }>0$ and $n_0 \in {\ 
\mathbb{N}}$ such that if $\left| \tau_n(\lambda) - \tau_n(\lambda^{\prime
}) \right| \leq R$ for \emph{some} $n \geq n_0$ and some $\lambda,
\lambda^{\prime }\in \Lambda$, then $\left| \lambda-\lambda^{\prime }\right|
\leq R^{\prime }$.
\end{definition}

The following results shows that the frame property of a Gabor system is
stable under Lipschitz deformations.

\begin{theorem}[{\protect\cite[Thm.~7.1 and Rem.~ 7.3]{grorro15}}]
\label{th_def_frames} Let $g \in M^1({\mathbb{R}^d})$ and $\Lambda \subseteq 
{\mathbb{R}^{2d}}$. Assume that  $\mathcal{G}(g,\Lambda)$ is a (Gabor)
frame and that $\Lambda_n \xrightarrow{Lip} \Lambda$.  Then  there exist $A,B>0$
and $n_0 \in {\mathbb{N}}$ such that $\mathcal{G}(g,\Lambda_n)$  is a frame with uniform  bounds $A,B$ for
all $n \geq n_0$. 
\end{theorem}

We will also need the following technical lemma concerning Lipschitz
convergence and relative separation.

\begin{lemma}[{\protect\cite[Lemma 6.7]{grorro15}}]
\label{lemma_sep} Let $\Lambda_n \xrightarrow{Lip} \Lambda$ and assume that $\Lambda$ is relatively separated. Then 
$\limsup_n \mathop{\mathrm{rel}}
(\Lambda_n) < \infty$.
\end{lemma}

The following corollary enables us to combine the stability of Gabor frames
under deformations of $\Lambda$ with small perturbations of the window $g$
on $M^1$-norm.

\begin{coro}
\label{coro_stabil} Assume that $g_n \longrightarrow g$ in $M^1({\mathbb{R}^d })$ and that $\Lambda_n \xrightarrow{Lip} 
\Lambda$. Then $\mathcal{G}
(g_n,\Lambda_n)$ is a frame for all sufficiently large $n$. (Moreover, the
corresponding frame bounds can be taken to be uniform in $n$).
\end{coro}

\begin{proof}
By Theorem \ref{th_def_frames}, there exist $A,B>0$ and $n_0 \in {\mathbb{N}}
$ such that for all $n \geq n_0$ 
\begin{align}  \label{eq_AB}
A \lVert f\rVert_2 \leq \lVert V_g f| \Lambda_n\rVert_2 \leq B \lVert
f\rVert_2, \qquad f \in {L^2({\mathbb{R}^d})}.
\end{align}
(Here $A,B$ are the square roots of the frame bounds.) By Lemma \ref{lemma_rel_sep}, $\Lambda$ and all $\Lambda_n$ with 
$n\gg 0$ are relatively
separated. Using Proposition \ref{prop_win_stab} we deduce that for all $f
\in {L^2({\mathbb{R}^d})}$ 
\begin{align*}
&\Big| \| V_g f| \Lambda_n\|_2- \| V_{g_n} f| \Lambda_n\|_2
\Big| \leq \lVert V_{g-g_n} f| \Lambda_n\rVert_2 \\
&\qquad \leq C \lVert g-g_n\rVert_{M^1} \mathop{\mathrm{rel}}(\Lambda_n)
\lVert f\rVert_2.
\end{align*}
Letting 
\begin{align*}
&A_n := A - C \lVert g-g_n\rVert_{M^1} \mathop{\mathrm{rel}}(\Lambda_n), \\
&B_n := B + C \lVert g-g_n\rVert_{M^1} \mathop{\mathrm{rel}}(\Lambda_n),
\end{align*}
we deduce from \eqref{eq_AB} and the triangle inequality that 
\begin{align}  \label{eq_ABn}
A_n \lVert f\rVert_2 \leq \lVert V_{g_n} f| \Lambda_n\rVert_2 \leq B_n
\lVert f\rVert_2, \qquad f \in {L^2({\mathbb{R}^d})}.
\end{align}
By Lemma \ref{lemma_sep} and the fact that $g_n \longrightarrow g$ in $M^1$
it follows that $A_n \longrightarrow A$ and $B_n \longrightarrow B$.
Combining this with \eqref{eq_ABn} we conclude that for all sufficiently
large $n$ 
\begin{align*}
A/2 \lVert f\rVert_2 \leq \lVert V_{g_n} f| \Lambda_n\rVert_2 \leq B/2
\lVert f\rVert_2, \qquad f \in {L^2({\mathbb{R}^d})}.
\end{align*}
Hence, for $n \gg 1$, $\mathcal{G}(g_n,\Lambda_n)$ is a frame with bounds $A^2/4, B^2/4$.
\end{proof}

\subsection{Flows and Lipschitz convergence}

A function $F: {\mathbb{R}} \times {\mathbb{R}^d} \to {\mathbb{R}^d}$ is
Lip\-schitz in the second variable if there exists $L>0$ such that 
\begin{align*}
\left| F(t,x)-F(t,y) \right| \leq L \left| x-y \right| , \mbox{ for all }
(t,x) \in {\mathbb{R}}\times{\mathbb{R}^d}.
\end{align*}
Under this assumption, we let $(\Phi_t)_{t\in {\mathbb{R}}} $ denote the
flow of $F$ (associated with time 0). This means that for each $x \in {\mathbb{R}^d}$, ${\mathbb{R}} \ni t \mapsto 
\Phi_t(x) \in {\mathbb{R}^d}$ is
a $C^1$ function and that

\begin{itemize}
\item[(a)] $\Phi_0(x)=x$,

\item[(b)] $\frac{d}{dt}\Phi_t(x)=F(t,\Phi_t(x))$.
\end{itemize}

The theory of ODEs implies that the flow exists and it is uniquely
determined by properties $(a)$ and $(b)$ above. Moreover, the flow satisfies
the following distortion estimate: given $T>0$, there exist constants $c_t,
C_T>0$ such that 
\begin{align}  \label{eq_dist}
c_T \left| x-y \right|   \leq \left| \Phi_t(x)-\Phi _t(y) \right|   \leq C_T
\left| x-y \right|  , \qquad x,y \in {\mathbb{R}}, \quad t \in [-T,T].
\end{align}
The previous estimate is normally proved using the following useful lemma.

\begin{lemma}[Gronwall]
\label{lemma_gron} Let $I \subseteq {\mathbb{R}}$ be an interval and $a \in I
$. Let $g: I \to [0,+\infty)$ be a continuous function that satisfies 
\begin{align*}
g(t) \leq A + B\left| \int_a^t g(s) ds \right| , \qquad t \in I,
\end{align*}
for some constants $A,B \in {\mathbb{R}}$. Then 
\begin{align*}
g(t) \leq Ae^{B\left| t-a \right| }, \qquad t \in I.
\end{align*}
(The reason for the absolute value outside the integral is that $t-a$ can be
negative.)
\end{lemma}

We now show that the flows of ODEs provide examples of Lipschitz
deformations.

\begin{theorem}
\label{th_lip_flow} Let $F: {\mathbb{R}} \times {\mathbb{R}^d} \to {\mathbb{R }^d}$ be Lipschitz in the second variable 
and let $(\Phi_t)_{t\in {\mathbb{R} }}$ be the corresponding flow. Let $\Lambda \subseteq {\mathbb{R}^d}$ be a
relatively separated set. Then $\Phi_t(\Lambda) \xrightarrow{Lip} \Lambda$,
as $t \longrightarrow 0$. \\
(More precisely, for each sequence $t_n \longrightarrow 0$, $\Phi_{t_n}(\Lambda) \xrightarrow{Lip} \Lambda$.)
\end{theorem}

\begin{proof}
Let $L>0$ be the Lipschitz constant of $F$ (in the second variable). We
first check condition $(L1)$ from Definition \ref{def_lip}. From the
definition of the flow it follows that 
\begin{align*}
\Phi_t(x) = x + \int_0^t F(s,\Phi_s(x)) ds,\qquad t \in {\mathbb{R}}.
\end{align*}
Therefore, 
\begin{align}
\Phi_t(\lambda)-\Phi_t(\lambda^{\prime })-(\lambda - \lambda^{\prime }) &=
\int_0^t \left( F(s,\Phi_s(\lambda)) - F(s,\Phi_s(\lambda^{\prime }))
\right) ds.
\end{align}
As a consequence, 
\begin{align*}
&\left| \Phi_t(\lambda)-\Phi_t(\lambda^{\prime })-(\lambda - \lambda^{\prime
}) \right| \leq \left| \int_0^t \left| F(s,\Phi_s(\lambda)) -
F(s,\Phi_s(\lambda^{\prime })) \right| ds \right| \\
&\qquad \leq \left| \int_0^t L \left| \Phi_s(\lambda)-\Phi_s(\lambda^{\prime
}) \right| ds \right| \\
&\qquad \leq L \left| \int_0^t \left| \Phi_t(\lambda)-\Phi_t(\lambda^{\prime
})-(\lambda - \lambda^{\prime }) \right| ds \right| + L\left| t \right|
\left| \lambda-\lambda^{\prime }\right| .
\end{align*}

Applying Gronwall's Lemma \ref{lemma_gron} to $g(t):=\left\vert \Phi
_{t}(\lambda )-\Phi _{t}(\lambda ^{\prime })-(\lambda -\lambda ^{\prime
})\right\vert $ we deduce that 
\begin{equation*}
\left\vert \Phi _{t}(\lambda )-\Phi _{t}(\lambda ^{\prime })-(\lambda
-\lambda ^{\prime })\right\vert \leq L\left\vert t\right\vert \left\vert
\lambda -\lambda ^{\prime }\right\vert e^{L\left\vert t\right\vert }.
\end{equation*}
Condition $(L1)$ follows from here.

To check condition $(L2)$, we consider only $t \in [-1,1]$ and \eqref{eq_dist} to obtain a constant $C$ such that 
\begin{align*}
C^{-1} \left| x-y \right| \leq \left| \Phi_t(x)-\Phi_t(y) \right| \leq C
\left| x-y \right| , \qquad t \in (-1,1).
\end{align*}
Hence, if for some instant $t_0$ we know that $\left|
\Phi_{t_0}(\lambda)-\Phi_{t_0}(\lambda^{\prime }) \right| \leq R$, then we
can deduce that $\left| \lambda-\lambda^{\prime }\right| \leq R^{\prime
}:=CR $.

This completes the proof.
\end{proof}

\section{Hamiltonian deformations: d\'{e}nouement}

\label{sec_de} We finally combine all  tools from the previous section
and prove the main result.

\begin{proof}[Proof of Theorem \protect\ref{th_small_time}]
Let us define $F: {\mathbb{R}} \times {\mathbb{R}^{2d}} \to {\mathbb{R}^{2d}}
$ by 
\begin{align*}
F(t,x,p):=(\partial_p \mathcal{H}(x,p), -\partial_x \mathcal{H}(x,p)).
\end{align*}
Then $F$ is a $C^1$ function with bounded derivatives  and,
consequently, $F$ is 
Lipschitz in the second set of  variables $(x,p)$. Let $t_n \longrightarrow 0$ and define $\Lambda_n := 
\Phi_{t_n}(\Lambda)$. Theorem \ref{th_lip_flow} implies that $\Lambda_n \xrightarrow{Lip} \Lambda$, while Theorem 
\ref{th_ham_win} implies
that $e^{-it_nH}g \longrightarrow g$ in $M^1$. Hence, Corollary \ref{coro_stabil} yields the desired conclusion.
\end{proof}

\begin{rem}
 The proof shows that, under  the conditions of Theorem
 \ref{th_small_time}, the 
Gabor systems $\mathcal{G}(e^{-ith}g, \Phi_t(\Lambda))$ admit uniform frame
bounds for $t \in [-t_0,t_0]$.
\end{rem}

\begin{rem}
We do not know whether the conclusion of Theorem \ref{th_small_time} remains
valid for arbitrary times. Moreover, we do not know of any example of a
Hamiltonian deformation that does not preserve the frame
property.
\end{rem}

\end{document}